\documentclass[11pt]{amsart}
\usepackage[margin=3.1cm]{geometry}
\usepackage{mathtools} 
\usepackage{amsthm,amsfonts,amssymb}
\usepackage{bm} 
\usepackage{hyperref}
\usepackage{todonotes}
\usepackage{tikz,color}
\usepackage{pgfplots}
\usepackage{pgffor}
\usetikzlibrary{shadings}
\definecolor{refkey}{rgb}{0.4,1,0.4}
\definecolor{labelkey}{rgb}{0.5,0.5,0.8}

\usetikzlibrary{angles,quotes,calc,shapes}
\newcommand{\RN}[1]{%
  \textup{\uppercase\expandafter{\romannumeral#1}}%
}

\usetikzlibrary{decorations.markings}
\usetikzlibrary{angles,quotes,calc}
\usetikzlibrary{calc}
 \usepackage{subfigure}

\pgfdeclareverticalshading{rainbow}{100bp}
{color(0bp)=(red); color(25bp)=(red); color(35bp)=(yellow);
color(45bp)=(green); color(55bp)=(cyan); color(65bp)=(blue);
color(75bp)=(violet); color(100bp)=(violet)}

\definecolor{viridis-01}{HTML}{440154}
\definecolor{viridis-02}{HTML}{482173}
\definecolor{viridis-03}{HTML}{433E85}
\definecolor{viridis-04}{HTML}{38598C}
\definecolor{viridis-05}{HTML}{2D708E}
\definecolor{viridis-06}{HTML}{25858E}
\definecolor{viridis-07}{HTML}{1E9B8A}
\definecolor{viridis-08}{HTML}{2BB07F}
\definecolor{viridis-09}{HTML}{51C56A}
\definecolor{viridis-10}{HTML}{85D54A}
\definecolor{viridis-11}{HTML}{C2DF23}
\definecolor{viridis-12}{HTML}{FDE725}

\hypersetup{colorlinks=true,allcolors=black}

\newtheorem{thm}{Theorem}

\newtheorem{lemma}{Lemma}
\newtheorem{cor}{Corollary}
\newtheorem{prop}{Proposition}

\newtheorem{remark}{Remark}
\newtheorem{defn}{Definition}
\newtheorem{example}{Example}

\newcommand{\Z}{\ensuremath{\mathbb{Z}}}
\newcommand{\R}{\ensuremath{\mathbb{R}}}

\newcommand{\cL}{\ensuremath{\mathcal{L}}}

\newcommand{\mca}{\operatorname{MCA}} 

\definecolor{pink}{rgb}{1,0,1}

\definecolor{mil}{rgb}{0.14,0.4,0.14}

\title[The strength of diversity]{The strength of diversity:  \\ mathematical proof that collections of variable individuals are robust in the face of challenges}  
\author{Julie Rowlett, Carl-Joar Karlsson \& Medet Nursultanov}
\date{}
\begin{document} 
\maketitle

\begin{quote} 
\textbf{Abstract.}  
Can one demonstrate quantitative effects of diversity within a system comprised of distinct individuals on the performance of the system as a whole?  Assuming that individuals can be different, we develop a model to interpolate between individual-level interactions and collective-level ramifications.  Rooted in theoretical mathematics, the model is not constrained to any specific context.  Potential applications include research, education, sports, politics, ecology, agriculture, algorithms, and finance.  Our first main contribution is a game theoretic framework for further analysis of the internal composition of an ensemble of individuals and the repercussions for the ensemble as a whole in competition with others. The second main contribution is the rigorous identification of all equilibrium points and strategies.  These equilibria suggest a mechanistic underpinning for biological and physical systems to tend towards increasing complexity and entropy, because diversity imparts strength to a system in competition with others.  
\end{quote}

\setcounter{tocdepth}{2} 

\section{Introduction} \label{s:intro} 
Diversity is an ubiquitous concept that appears in multiple fields including scientific research \cite{freeman2014collaboration}, education \cite{nivet2011commentary,milem2003educational,smith2000benefits}, human resource management \cite{roberge2010recognizing,kearney2009when}, business \cite{fine2020does,friedman2013biculturalism}, sports \cite{cunningham2011benefits},  politics \cite{bohman2006deliberative}, ecology \cite{yang2018sustainable,redlich2018landscape,diaz2011linking,grime1998benefits,isbell2017benefits,jrsi, flock2018}, agriculture \cite{indiacrop_2016, chinacrops_2014, lin_crops2011, canada_crops2004}, algorithms \cite{sudholt2020benefits}, power grids \cite{power21}, networks \cite{networks21}, and finance \cite{gambarelli2004takeover,chis2017coalitional,simonian2019portfolio,amihud1974portfolio,bell1988game,young1998minimax,yang2013multi,fu2017information, veysoglu2002thesis}.  In this work, diversity is a flexible concept that broadly means \em anything \em that \em differentiates.  \em  For groups of people, this includes not only demographic differences \cite{freeman2014collaboration} but also deep diversity like personality, mentality, and past experiences  \cite{norway_deepdiverse}.  In biology, diversity can be measured at different levels ranging from an ecosystem comprised of different species \cite{Xu_2020} to a single species comprised of genetically identical yet phenotypically variable individuals \cite{plos}.  In finance, diversity may simply refer to distinct investment products such as stocks, mutual funds, or bonds \cite{markowitz1952, markowitz1959}.  What role does diversity play in these contexts?  Is there a unifying mechanism that applies to this myriad of seemingly unrelated disciplines?  To address these questions, we invoke theoretical mathematics, allowing diversity to take on specific interpretations in specific applications.  In this way, we may discover a quantitative linkage between diversity within an ensemble and the effect on competition with other ensembles.  

In research, according to \cite{freeman2014collaboration}, ``a paper generated by a more diverse research group could tap into different networks and thus attract greater attention and citations, an effect observed in patents studies \cite{kerr2008revecon}, and in inter-institution and international collaborations \cite{adams_fourthage2013}.''  A larger research group may also be more likely to present a correct analysis and to draw reliable conclusions if all group members contribute to a rigorous internal review process.  Quoting \cite{diversitychallenge2014} ``There is growing evidence that embracing diversity — in all its senses — is key to doing good science.''   In a quantitative study of over 9 million papers and 6 million scientists \cite{preeminence2018}, the authors found ``that ethnic diversity resulted in an impact gain of 10.63 \% for papers, and 47.67\% for scientists.''  

One of the reasons diversity may be beneficial in research is that teams with members from diverse backgrounds may have a greater variety of perspectives \cite{freeman2014collaboration}.  For the same reason diversity may be broadly beneficial in education.  Quoting  \cite{educationreport_2016} ``researchers have documented that students’ exposure to other students who are different from themselves and the novel ideas and challenges that such exposure brings leads to improved cognitive skills, including critical thinking and problem solving.''  Here we observe that  ``different from themselves'' may be not only due to demographic variables, but also deep level diversity.  The benefits of deep level diversity were explored in a corporate context \cite{norway_deepdiverse} investigating 385 Norwegian companies.  The aforementioned study provided strong support for the notion that the higher the level of board diversity with respect to the board members’ backgrounds (both professional and personal) and personalities, the higher the degree of board creativity and cognitive conflict during the decision-making process.  The deep level diversity of board members may result in a board that possesses a greater set of skills, competencies, and perspectives.  In the context of professional sports, meta analysis showed that overall group diversity has a positive effect on group outcomes \cite{eu_sports2018}.  Consequently, a sports team with a broad skill set may be able to outcompete a team with a narrow skill set by exploiting those skills which are lacking in a team with less diversity across its members.

In biology, the overall health of an ecological system is often judged by its level of biodiversity \cite{Xu_2020}.  One instance of tremendous biodiversity is provided by marine microbes.  Their species diversity is estimated to exceed 200,000 species in the plankton \cite{devargas2015, sunagawa2015}.  At all levels of taxonomy, from species to intra-strain comparisons, there exists a tremendous variability in genetic, physiological, behavioral and morphological characteristics \cite{ahlgrenrocap, ahlgren2006, johnson2006, schaum2012, boyd2013, hutchins2013, kashtan2014, harvey2015, mdmontalbano, sohm2016, godherynearson, wolf2017, olofsson2018, falkowski2008, worden2015}.  In \cite{jrsi, thee, plos}, we suggested that this phenotypic heterogeneity in all microbe organisms is what makes it possible for countless microbe species to coexist and for new species to continually emerge \cite{eatingplastic}.  Another microbial ecological system that recently has garnered much attention is the human microbiome in general, and the gut microbiome in particular.  The gut bacterial ecosystem is important for health, not only digestive and metabolic function, but also cardiovascular and neuropsychiatric health \cite{Xu_2020}.  Public databases estimate on the order of 10,000 bacterial species in this ecosystem \cite{gutbacteria}.  Reduced gut biodiversity is associated with health impairment such as Crohn's disease \cite{gutcrohn}.  In some sense, just as the diversity of a research team may contribute to their ability to create and innovate, biodiversity may play a similar role for ecological systems by facilitating ecological innovation.  The richer the diversity of life, the greater the opportunity for medical discoveries, economic development, and adaptive responses to new challenges.  

In finance, a cornerstone of modern investment strategies, developed by Harry Markowitz in the 1950s \cite{markowitz1952, markowitz1959} is known as \em modern portfolio theory, \em for which Markowitz received the Nobel Prize in Economics in 1990.   The prize recognized his development of a rigorous operational theory for portfolio selection under uncertainty which has evolved into a foundation for financial economics research.  A key concept in modern portfolio theory is to simultaneously analyze two dimensions:  the expected return on the portfolio and its variance.   Based on Markowitz's work, an investor can construct a portfolio of multiple assets to maximize returns for a given level of risk. Conversely, given a desired level of expected returns, the investor can construct a portfolio with the lowest possible risk.  Although it may seem unrelated, a similar approach towards risk mitigation has been suggested in agriculture \cite{indiacrop_2016, chinacrops_2014, lin_crops2011, canada_crops2004}.  Analogous to portfolio diversification, crop diversification may become increasingly important the context of climate change.  

In all of these contexts, diversity is beneficial for reasons that may be loosely connected, but without one clear mechanistic underpinning.  To investigate the strength of diversity in a sufficiently broad sense that encompasses all of these contexts, we turn to theoretical mathematics.  The advantage is that theoretical mathematics is not constrained to any one specific application.  The limitation is that simplifying assumptions must be made to obtain results, and so a theoretical mathematical model will never be a perfect real-world match.  However, the same holds for all fundamental science, and one cannot deny its utility.  In the aforementioned contexts, a collection of individual entities comprise a team that competes with other teams.  All teams must obey a set of rules.  Game theory sets a natural mathematical foundation to analyze such situations. 

Harnessing the tools of game theory requires a mechanism for interpolating between the individuals that comprise a team and the interactions between different teams.  Historically, many authors have modeled competing teams as single players in the game theoretic sense \cite{liang2012game,novak2010differential,lye2004game,jormakka2005modelling}.  Quoting \cite{bornstein1997cooperation}, the ``use of a two-person game to model conflict between groups presupposes that all group members have identical preferences over the set of possible outcomes and therefore that each group can be treated as a unitary player.''  In biology, it has also been common to analyze competition between species by viewing the species as the player in the game theoretic sense \cite{barabas2016}.  Those approaches provide no mechanism to interpolate between the possibly diverse individuals of a team and the repercussions of the internal composition of the team for its competition with other teams.  To investigate teams comprised of unique and possibly diverse individuals, we introduce a mathematical model that quantifies how the composition of a team affects its competition with other teams.  Our first main contribution thereby provides a mechanism that interpolates between the individuals comprising a team and the interactions between different teams.  This can be used to analyze the consequences of the internal team composition for competitions between different teams.  As a result of our analysis, our second main contribution is mechanistic underpinning for the strengths of diversity, especially in situations in which new or unpredictable challenges occur.  

\section{Results} \label{s:results} 
The games of teams we introduce here generalize the game theoretic competitive model Rowlett et. al. introduced in \cite{jrsi, thee, plos}.  In \cite{thee, plos}  a major aim was to interpolate from individual competitions between microbes to the cumulative consequences for the species.  However, there is no mathematical reason that the individual competitors in that model must be microbes, or anything else for that matter.  One of the strengths of theoretical mathematics is that it is not constrained to specific applications.  Consequently, a model developed with one application in mind may prove useful for numerous other contexts. The game theoretic framework we construct here is a significant generalization of the model developed in \cite{thee, plos}.  We consider the game theoretic framework to be a meaningful contribution and therefore in itself a result because it is a tool that can be applied to any collection of individuals that compete with other collections of individuals, whether they are people, animals, microbes, investment products, or anything else.  This foundation can be built upon to include further specific details and particulars depending on the specific context of interest.  

\subsection{Games of teams} \label{s:setup}
A real value $x$ represents a \em competitive ability, \em in the sense that $x$ is superior to all values that are strictly smaller, whereas $x$ is inferior to all values that are strictly greater.  Simply put, $x$ beats anything lower and is beaten by anything higher; the same value is a tie.  The competitive ability is a versatile concept that can be adapted to each specific field of application.  It could be used to quantify one specific characteristic that is pertinent to competition, or it could be used to represent an aggregate assessment across all competition relevant characteristics.  The competitive ability could also be used to represent resource allocation within a team.   A \em strategy \em is a rule for assigning competitive abilities to the individuals that comprise the team subject to a constraint that may correspond to biological or financial limitations.  We may at times abuse notation by identifying a team with its strategy.  The strategy is a rule for assigning the competitive abilities of the team members, but we note that this does not mean that each team member's competitive ability is constant over time.  The competitive abilities of the individuals can vary while maintaining a given strategy for the team as a whole. We would like to understand what is the best strategy subject to a constraint on the \em mean competitive ability \em assessed over all individuals comprising a team.  Subject to a constraint on the mean competitive ability, what is the best way to assign competitive abilities to the individuals of a team?  Equivalently, what is the best way to allocate resources to the members of a team, subject to a constraint on the total amount of resources available?  

\begin{defn} \label{def:strategy} 
Let $f$ be a non-negative, bounded, Lebesgue measurable function that is not identically zero and whose support is a compact set contained in $[a, \infty)$, for some fixed $a \in \R.$  Let 
\begin{equation}
\label{eq:F}
F(x):=\int_a^x f(t)dt = \int_{[a,x]} f d\mu, \quad ||f||_{L^1} = \int_a ^\infty f d\mu
\end{equation}
with integration respect to the Lebesgue measure, $\mu$.  All strategies will be assumed to satisfy the constraint on the \em mean competitive ability:  \em  
\begin{equation}\label{eq:constraint}
\mca(f) := \frac{1}{||f||_{L^1}} \int_a ^\infty tf(t) dt \leq C, \textrm{ for a fixed } C > a. 
\end{equation}

Then, such a function $f$ is known as an \em $\cL^\infty$ strategy, \em or equivalently, as a \em bounded measurable strategy.  \em If we further assume that $f$ must be continuous, then it is a \em continuous strategy.  \em  The corresponding competitive games are known as the \em bounded measurable game of teams \em and the \em continuous game of teams, \em respectively.  We will also analyze \em discrete strategies.  \em  A discrete strategy assigns a non-negative value to each element in a  discrete set of values 
\[ \left\{ \frac j M \right\}_{j\in \Z, j \geq a}, \quad x_j := \frac j M.\] 
We note here that $a \in \Z$ and $M>0$ are fixed.  
A \em discrete strategy \em is therefore identified with a map 
\[ A : \left\{ \frac j M \right\}_{j\geq a} \to [0, \infty), \quad |A| := \sum_{j\geq a} A(j/M) > 0.\] 
We will assume that the support of the strategy is finite, analogous to the assumption in the bounded measurable and continuous cases that functions have compact support.  This immediately implies that $|A|$ is finite.  
In this case, we assume that all strategies satisfy the following constraint on the \em mean competitive ability:  \em  
\[ \mca(A) :=  \frac{1}{|A|} \sum_{j\geq a} A(j/M) \frac j M \leq C, \textrm{ for a fixed } C > \frac a M.\] 
The game in this case is the \em discrete game of teams.  \em  
\end{defn} 

We suggest that it is reasonable to assume that strategies are compactly supported, because in all practical applications of which we are aware, this will always be the case.  The parameters $a$, $C$, and $M$ are completely free to choose and customize according to the specific field of application.
 
\subsubsection{Team payoffs and equilibrium strategies} \label{ss:payoffs}    
As in \cite{thee, plos}, the individuals from the teams are randomly paired to compete.  For a collection of competing teams $\{f_k\}_{k=1} ^n$ in the bounded measurable and continuous games, we define the payoff to strategy $f_k$ by assessing the cumulative wins and losses of all individuals 

\begin{equation} \label{eq:def_payoff_c}  \wp (f_k; f_1, \ldots, f_{k-1}, f_{k+1}, \ldots, f_n) := 
\int_a ^\infty f_k (x) \left[ \int_a ^x \sum_{\ell \neq k}  f_\ell (t) - \int_x ^\infty \sum_{\ell\neq k} f_\ell (t) \right] dx. \end{equation}

For the discrete game, the payoffs are derived and defined analogously, so for a collection of competing teams, the payoff to strategy $A_k$ is 
\begin{equation}\label{eq:def_payoff_d} \wp (A_k; A_1, \ldots, A_{k-1}, A_{k+1}, A_n) :=  
 \sum_{j\geq a} A_k (x_j) \left [ \sum_{a \leq i<j} \sum_{\ell \neq k} A_\ell (x_i) - \sum_{i>j} \sum_{\ell \neq k}  A_\ell (x_i) \right ]
\end{equation} 

Whenever a sum is empty, it is defined to be zero.  An important notion in game theory is an \em equilibrium point, \em also known as a \em Nash equilibrium point \em due to Nash's proof of their existence \cite{nash}. This is a collection of strategies for all competing teams so that if any one team alone changes their strategy, their payoff does not increase.  

\begin{defn} \label{defn:eq_p_str} 
For $n$ competing teams, an \em equilibrium point \em consists of $n$ strategies for the $n$ teams that satisfy the following condition:  For each $k=1, \ldots, n$, if Team $k$ changes its strategy but all other Teams $\ell$ for all $\ell \neq k$ retain their strategies, then the payoff to Team $k$ does not increase.  That is to say, for all $k=1, \ldots, n$ we have in the bounded measurable and continuous games 
\[ \wp (f_k; f_1, \ldots, f_{k-1}, f_{k+1}, \ldots, f_n) \geq \wp (g; f_1, \ldots, f_{k-1}, f_{k+1}, \ldots, f_n) \] 
for any strategy $g$ of the same type (bounded measurable or continuous).  In the discrete game of teams, similarly, the strategies must satisfy 
\[   \wp (A_k; A_1, \ldots, A_{k-1}, A_{k+1}, A_n) \geq \wp(B; A_1, \ldots, A_{k-1}, A_{k+1}, A_n),\] 
for all $k=1, \ldots, n$ and for all discrete strategies $B$.  The strategies that comprise an equilibrium point are known as \em equilibrium strategies.  \em  
\end{defn} 

In many contexts, it is reasonable to expect that the `best' strategy for all players considered simultaneously are those in an equilibrium point \cite{nash}.  For this reason, our results here identify all equilibrium strategies and all equilibrium points.  

\subsection{Teams characterized by equilibrium strategies are those with diverse individuals}
In addition to developing the game theoretic framework, that could be considered a result on its own, we identify all equilibrium points and equilibrium strategies for these games of teams.  A visualization of equilibrium strategies for the bounded measurable game is shown in Figure \ref{fig:eq_bmg}, but we note that these are just finitely many examples of the infinitely many equilibrium strategies.  Similarly, in the discrete game there are also infinitely many equilibrium strategies with examples shown in Figure \ref{fig:eq_dg}. 

\begin{figure} \centering
\begin{tikzpicture}
	\draw[ultra thick, red] (0,0)--(1,0) node[black, below]{$a$}; 
	\draw[thick, red, dashed] (1,0) -- (1, 3.5); 
	
	\draw[ultra thick, red] (1,3.5)--(4.6,3.5) node[red, right]{$f(x)$};
	\draw[ultra thick, red] (4.6,0.01)--(7.6,0.01);
	\draw[thick, red, dashed] (4.6,0.01)--(4.6,3.5) node[black,below=101pt]{$2C-a$};
	\draw[ultra thick, orange] (0,0) -- (1, 0); 
	\draw[thick, dashed, orange] (1, 0) -- (1, 3); 
	\draw[ultra thick, orange] (1,3)--(4.6,3)node[orange, right]{$f(x)$};
	\draw[ultra thick, orange] (4.6,0.01)--(7.6,0.01);
	\draw[thick, orange, dashed] (4.6,0.01)--(4.6,3);
	\draw[ultra thick, yellow] (0,0) -- (1, 0); 
	\draw[thick, dashed, yellow] (1, 0) -- (1, 2.5); 
	\draw[ultra thick, yellow] (1,2.5)--(4.6,2.5)node[yellow, right]{$f(x)$};
	\draw[ultra thick, yellow] (4.6,0.01)--(7.6,0.01);
	\draw[thick, yellow, dashed] (4.6,0.01)--(4.6,2.5);
	\draw[ultra thick, green] (0,0) -- (1, 0); 
	\draw[thick, dashed, green] (1, 0) -- (1, 2); 
	\draw[ultra thick, green] (1,2)--(4.6,2)node[green, right]{$f(x)$};
	\draw[ultra thick, green] (4.6,0.01)--(7.6,0.01);
	\draw[thick, green, dashed] (4.6,0.01)--(4.6,2);
	\draw[ultra thick, cyan] (0,0) -- (1, 0); 
	\draw[thick, dashed, cyan] (1, 0) -- (1, 1.5); 
	\draw[ultra thick, cyan] (1,1.5)--(4.6,1.5)node[cyan, right]{$f(x)$};
	\draw[ultra thick, cyan] (4.6,0.01)--(7.6,0.01);
	\draw[thick, cyan, dashed] (4.6,0.01)--(4.6,1.5);
	\draw[ultra thick, blue] (0,0) -- (1, 0); 
	\draw[thick, dashed, blue] (1, 0) -- (1, 1); 
	\draw[ultra thick, blue] (1,1)--(4.6,1)node[blue, right]{$f(x)$};
	\draw[ultra thick, blue] (4.6,0.01)--(7.6,0.01);
	\draw[thick, blue, dashed] (4.6,0.01)--(4.6,1);
	\draw[ultra thick, violet] (0,0) -- (1, 0); 
	\draw[thick, dashed, violet] (1, 0) -- (1, 0.5); 
	\draw[ultra thick, violet] (1,0.5)--(4.6,0.5)node[violet, right]{$f(x)$};
	\draw[ultra thick, violet] (4.6,0.01)--(7.6,0.01);
	\draw[thick, violet, dashed] (4.6,0.01)--(4.6,0.5);
	\draw[thick, ->] (0,-.5)--(0,4) node[left]{$y$};
	\draw[thick, ->] (-.5,0)--(8,0) node[below]{$x$};
\end{tikzpicture}
\caption{These are examples of equilibrium strategies in the bounded measurable game of teams.  There are infinitely many equilibrium strategies, because any function that is constant and positive on $[a, 2C-a]$ and zero elsewhere is an equilibrium strategy.  Consequently, teams characterized by an equilibrium strategy span the whole range of diverse competitive abilities from the minimum value, $a$ up to twice the constraint value minus $a$.  Equivalently, a team characterized by an equilibrium strategy allocates resources evenly across all team members, centered around the constraint value.} 
\label{fig:eq_bmg}
\end{figure}
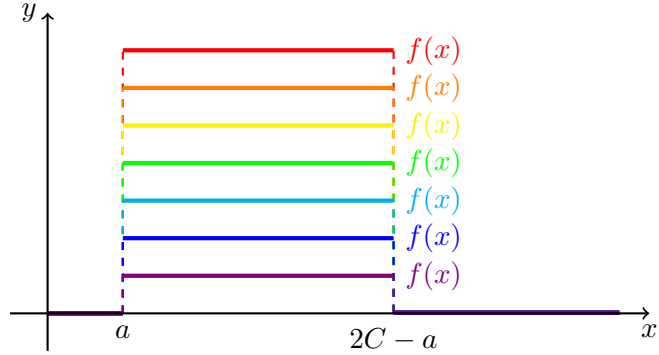

\begin{thm} \label{th:eq_st_bm} 
In the bounded measurable game, a strategy is an equilibrium strategy if and only if it is almost everywhere equal to 
\[ \begin{cases} c = \textrm{ constant } > 0 & \textrm{ on } [a, 2C-a] \\ 0 & \textrm{ on } (2C-a, \infty].\end{cases} \] 
Any collection of equilibrium strategies is an equilibrium point, and conversely, every equilibrium point is comprised of these equilibrium strategies.  The sum of two or more equilibrium strategies is an equilibrium strategy.  
\end{thm}

\begin{thm} \label{th:eq_st_c} 
In the continuous game, there are no equilibrium strategies. 
\end{thm}

\begin{remark}  The obstruction to the existence of equilibrium strategies for the continuous game is that the functions which should be equilibrium strategies are those in the bounded measurable game.  The problem with these strategies is that they are not continuous on $[a, \infty)$.  This is no longer a problem for the discrete game. 
\end{remark} 

\begin{thm} \label{th:eq_st_d} 
In the discrete case, assume that the set of competitive abilities is 
\[ \left\{ x_j = \frac j M \right\}_{j\geq a}, \textrm{ with constraint value $C=\frac{k+a}{2M}$, for an integer $k > 0$.}\]  
If $k+a$ is odd,  then $B$ is an equilibrium strategy if and only if it satisfies for some constant $c>0$, 
\[ \begin{cases} B(x_{j}) = c, & a \leq j \leq k, \\ B(x_{j}) = 0, & k < j. \end{cases} \] 
If $k+a$ is even, then $B$ is an equilibrium strategy if and only if $\mca(B) = C$, and  $B(x_{2j+a}) = B(x_a)$, and $B(x_{2j+a+1}) = B(x_{a+1})$ for $j=0, \ldots, k$, with $B(x_{j}) = 0$ for all $j>k$.  In all cases, any collection of equilibrium strategies is an equilibrium point, and conversely, every equilibrium point is comprised of these equilibrium strategies.  In all cases, the sum of two or more equilibrium strategies is an equilibrium strategy.
\end{thm} 

The equilibrium strategies are characterized by the distribution of competitive abilities spanning the range from $a$ to $2C-a$.  For a team of individual people or organisms, this could be interpreted as a maximally diverse team, because their competitive abilities are spread across this range.  For a team of investment products, this could be interpreted as a maximally diverse financial portfolio.  Any strategy that is not an equilibrium strategy can be defeated, in the sense that we provide a recipe in the proof to construct a strategy that will defeat any non-equilibrium strategy in competition.  

\begin{figure}[h]
	\centering%
	\begin{minipage}{0.45\textwidth}
		\centering
		\begin{tikzpicture}
			\node[fill=white] at (3.2,3.6) {constraint value $\frac{k+a}{2M}$ with $k+a$ even};
			\draw[->] (0,0)--(0,4) node[left]{$y$};
			\draw[->] (0,0)--(7,0) node[below]{$x_j$};
			\draw[fill=viridis-01] (-0.3,0.0) rectangle (0.3,3);
			\draw (0,0)node[black,below]{$x_a$};
			\draw[fill=viridis-02] (0.7,0.0) rectangle (1.3,2);
			\draw (1,0)node[black,below]{$x_{a+1}$};
			\draw[fill=viridis-04] (1.7,0.0) rectangle (2.3,3);
			\draw (2,0)node[black,below]{$x_{a+2}$};
			\draw[fill=viridis-06] (2.7,0.0) rectangle (3.3,2);
			\draw (3,0)node[black,below]{$x_{a+3}$};
			\draw[fill=viridis-08] (3.7,0.0) rectangle (4.3,3);
			\draw (4,0)node[black,below]{$x_{a+4}$};
			\draw[fill=viridis-10] (4.7,0.0) rectangle (5.3,2);
			\draw (5,0)node[black,below]{$x_{a+5}$};
			\draw[fill=viridis-12] (5.7,0.0) rectangle (6.3,3);
			\draw (6,0)node[black,below]{$x_{a+6}$};
		\end{tikzpicture}
	\end{minipage}%
	\hfill%
	\begin{minipage}{0.45\textwidth}
		\centering
		\begin{tikzpicture}
			\node[fill=white] at (3.2,3.6) {constraint value $\frac{k+a}{2M}$ with $k+a$ odd};
			\draw[->] (0,0)--(0,4) node[left]{$y$};
			\draw[->] (0,0)--(6,0) node[below]{$x_j$};
			\draw[fill=viridis-01] (-0.3,0.0) rectangle (0.3,2.5);
			\draw (0,0)node[black,below]{$x_a$};
			\draw[fill=viridis-03] (0.7,0.0) rectangle (1.3,2.5);
			\draw (1,0)node[black,below]{$x_{a+1}$};
			\draw[fill=viridis-05] (1.7,0.0) rectangle (2.3,2.5);
			\draw (2,0)node[black,below]{$x_{a+2}$};
			\draw[fill=viridis-08] (2.7,0.0) rectangle (3.3,2.5);
			\draw (3,0)node[black,below]{$x_{a+3}$};
			\draw[fill=viridis-10] (3.7,0.0) rectangle (4.3,2.5);
			\draw (4,0)node[black,below]{$x_{a+4}$};
			\draw[fill=viridis-12] (4.7,0.0) rectangle (5.3,2.5);
			\draw (5,0)node[black,below]{$x_{a+5}$};
		\end{tikzpicture}
	\end{minipage}%
\caption{These are examples of equilibrium strategies in the discrete game of teams.  There are infinitely many equilibrium strategies.  The common feature they all share is that these strategies always span the whole range of diverse competitive abilities from $a$ to $2C-a/M$.} 
\label{fig:eq_dg}
\end{figure}
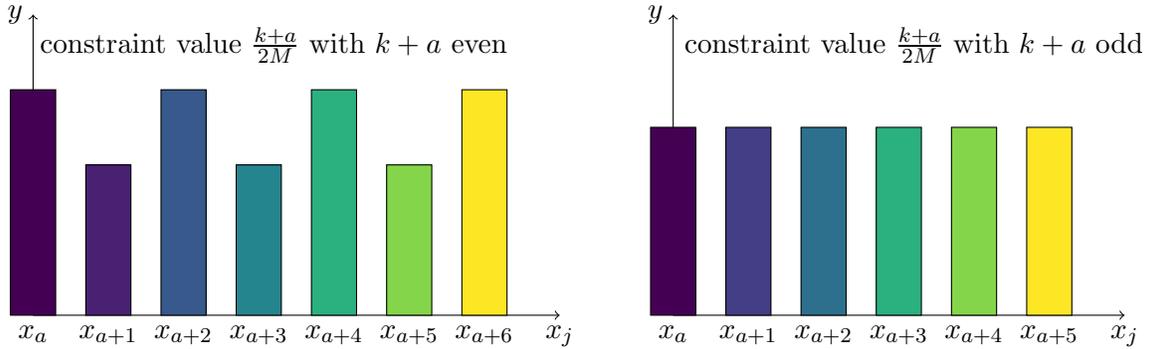

\subsection{Examples that fit with the theoretical predictions} \label{s:examples} 
The competitive ability is a flexible concept.  We propose that it should take into account traits and features that are relevant to the type of competition.  For example, in application to sports teams, a player could be assigned a competitive ability based on their cumulative statistics, which one would reasonably expect to be correlated to their salary.  In application to investment products, a competitive ability could be analogously defined based on past performance.  The competitive ability could be either focused on a specific aspect of competition, or it could be an aggregate assessed by considering several aspects of competition.  In our model, the competitive abilities of the individuals in a team do \em not \em need to be constant!  Considering several competitions, in one round, an individual may have higher ability, whereas in the next round, their ability might be lower.  Another interpretation of the distribution of competitive abilities within a team could be the distribution of resources within a team.  The equilibrium strategies distribute the resources equally across the range $[a, 2C-a]$.  If this distribution is perfectly random across individuals and is re-assigned randomly in each round of competition, then this is the most fair way to distribute resources across the team.  Equilibrium strategies have the interesting feature that they are neutral in competition towards each other; they peacefully co-exist.   At equilibrium, the teams are equally matched.  The payoffs in these games of teams represent the cumulative gains and losses of the team, so from a resource perspective, the teams are neither gaining nor losing their resources.  

\begin{example}  Consider a football (soccer) team.  The team has a certain budget for the players' monthly salaries.  Assume the minimum salary for a player is $20,000$ euros per month, and the team has a total budget of $400,000$ euros per month to spend on salaries.  Assume that one wishes to employ precisely 11 players.  Then, in this example, $a=20,000$ and $C=400,000/11$.  Then, according to the theorem, if we hire exactly 11 players, one player should be chosen whose salary is approximately $20,000$, a second player with salary approximately $23,636$, a third player with salary approximately $27,272$, and so forth, up to the most expensive player whose salary is approximately $52,727$.  Is this really the best way to choose a team?  It is interesting to note that the actual salaries of players in football teams are spread across a range quite similar to this prediction according to our equilibrium strategies; see for example \href{https://salarysport.com/football/bundesliga/}{bundesliga}. 
\end{example}

\begin{example} In the game rock-paper-scissors, the unique equilibrium strategy is the mixed strategy corresponding to drawing rock, paper, and scissors purely randomly and with equal probability.  In biology, this has been used to explain strains of bacteria in which some individuals produce a toxin (T), some are resistant to the toxin (R), and some are susceptible to the toxin (S) \cite{rps_trs}.  The analogue to rock-paper-scissors is through:  T beats S but loses to R, S beats R but loses to T, and R beats T but loses to S.  Both the RPS equilibrium strategy as well as the diverse individuals within the bacteria strain fit with our predictions.  Moreover, there are numerous examples of intraspecific variability within microbe strains and the positive effect on survival in adverse conditions \cite{jrsi, thee, plos, godherynearson, devargas2015, harvey2015, wolf2017}. 
\end{example}

\begin{example}
A balanced investment strategy also resembles an equilibrium strategy, for example, a portfolio consisting of 25 percent dividend paying blue chip stocks, 25 percent small capitalization stocks, 25 percent AAA rated government bonds, and 25 percent investment grade corporate bonds as described in \href{https://www.investopedia.com/terms/b/balancedinvestmentstrategy.asp}{balanced investment strategy}.  Identifying each of these four types of investments with a range of four different competitive abilities, an equilibrium strategies corresponds to investing equal amounts in each of the four types.  
\end{example} 

\section{Mathematical proofs} \label{s:map} 
We begin by proving that the games are translation invariant in a certain sense.  This allows us to reduce to the case in which all strategies are supported in $[0, 1]$, and the constraint value $C \leq \frac 1 2$.  Next, we demonstrate a sufficient condition for a strategy to be an equilibrium strategy.  It is sufficient that the payoff in competition with any other strategy is non-negative.  We then prove that the equilibrium strategies given in Theorems \ref{th:eq_st_bm} and \ref{th:eq_st_d} satisfy this condition by explicitly computing the payoffs according to their definitions.  Moreover we determine all equilibrium strategies in the case of two competing teams.  Finally, we complete the proof by demonstrating that the sufficient condition to be an equilibrium strategy is also necessary so in fact we locate all equilibrium strategies in this way. 

\subsection{Translation invariance} \label{ss:restrictions_not} 
If there is a collection of competing teams, then there is a bounded closed interval that contains all of their supports.  By possibly expanding the interval, we may assume that it is of the form $[a, b]$, and that the constraint value $C \in (a, b)$.  The following lemma shows that it is equivalent to assume the interval is $[0, 1]$.  

\begin{lemma} \label{le:translation_bmc}
Assume that $f$ and $g$ are both non-negative bounded measurable functions whose supports are contained in an interval $[a,b]$ for a bounded interval with $-\infty < a < b < \infty$.  Define $\ell := b-a$, and assume that $C \in (a,b)$.  Let
\[ F(x) := \int_a ^x f(t) dt, \quad x \in [a,b], \quad G(x) := \int_a ^x g(t) dt.\]
Assume that $f$ and $g$ are both not identically zero.  The payoff is defined to be 
\[ \wp[f; g] := \int_a^b f(x) \left(\int_a^x  g(t) dt-\int_x^b g(t) dt\right)dx.\] 
Then, define 
\[ \widetilde h (t) := h(t \ell + a) = h(x), \quad h \in \{ f, g\}, \quad \widetilde C := \frac{C-a}{\ell}.\] 
Then $f$ and $g$ satisfy the constraint 
\[   \int_a ^b x h(x) dx \leq C \int_a ^b h(x) dx \iff \int_0 ^1 t \widetilde h(t) dt \leq \widetilde C \int_0 ^1 \widetilde h(t) dt, \quad h \in \{ f, g \}. \] 
Moreover with the payoffs to $\widetilde f$ and $\widetilde g$ satisfy
\[ \wp [f;g] = \wp[\widetilde f; \widetilde g], \quad \wp[g;f] = \wp[\widetilde g; \widetilde f].\] 
\end{lemma} 

\begin{proof} 
We compute using the change of variables $\widetilde f(t) := f(t \ell + a) = f(x)$, 
\[ \int_a ^b x f(x) dx \leq C \int_a ^b f(x) dx \iff \int_0 ^1 t \widetilde f(t) dt \leq \widetilde C \int_0 ^1 \widetilde f(t) dt, \quad \widetilde C := \frac{C-a}{\ell}.\] 
Consequently, $f$ satisfies the constraint on the interval $(a,b)$ with constraint value $C \in (a,b)$ if and only if $\widetilde f$ satisfies the constraint on the interval $(0, 1)$ with constraint value $\widetilde C \in (0,1)$.  Moreover, for $g$ also satisfying the constraint on $(a,b)$, note that for $c \in [0, 1]$, with $\widetilde g(t) := g(t \ell + a) = g(x)$, 
\[ \widetilde G(c) := \int_0 ^c \widetilde g(t) dt = \frac 1 \ell \int_a ^{\ell c + a} g(x) dx = \frac{G(\ell c + a)}{\ell}.\] 
Consequently,
\[ \wp [f;g] = \int_a ^b f(x) (2 G(x) - G(b)) dx = \int_a ^b f(x) \left( \frac{2 G(x)}{\ell} - \frac{G(b)}{\ell} \right) \ell dx\]  
\[ = \int_0 ^1 \widetilde f(t) (2 \widetilde G(t) - \widetilde G(1))dt = \wp [\widetilde f; \widetilde g].\] 
\end{proof} 

No generality is therefore lost by assuming the competitive abilities are contained in the interval $[0, 1]$ for the bounded measurable game as well as the continuous game.  The following lemma shows that the same is true for the discrete game.  

\begin{lemma} \label{le:translation_disc} 
Assume that $A$ and $B$ are discrete strategies that define maps 
\[ A, B: \left\{ x_j := a + \frac{j \ell}{M} \right\}_{j \geq 0}  \to [0, \infty), \quad \ell := b-a, \quad -\infty < a < b < \infty,\] 
such that 
\[ |A| := \sum_{j\geq 0}  A(x_j) > 0, \quad \textrm{ and } |B| := \sum_{j\geq 0}  B(x_j) > 0.\] 
Assume that $C \in (a,b)$.  Then, 
\[ \mca(A) := \frac{1}{|A|} \sum_{j\geq0}  x_j A(x_j) \leq C \iff \mca(\widetilde A) \leq \widetilde C := \frac{C-a}{\ell},\] 
for 
\[ \widetilde A : \left \{ \frac j M \right\}_{j\geq 0}  \to [0, \infty), \quad \widetilde A(j/M) := A(x_j).\] 
Moreover for 
\[ \wp [A; B] := \sum_{j\geq 0}  A(x_j) \left( \sum_{i<j} B(x_i) - \sum_{i>j} B(x_i) \right),\] 
and $\widetilde B$ defined analogously to $\widetilde A$ 
we have 
\[ \wp [A;B] = \wp [\widetilde A; \widetilde B], \quad \wp [B;A] = \wp [\widetilde B; \widetilde A].\] 
\end{lemma} 

\begin{proof} 
Note that $|A| = |\widetilde A|$, and $|B| = |\widetilde B|$.  Then 
\[ \mca(\widetilde A) = \frac{1}{|\widetilde A|} \sum_{j\geq 0}  \frac j M \widetilde A(j/M) = \frac{1}{|A|} \sum_{j\geq 0}  \frac j M A(x_j) =  \frac{1}{|A|} \sum_{j\geq 0}  \frac{x_j - a}{\ell} A(x_j)\]  
\[ = \frac{\mca(A) - a}{ \ell} \leq \widetilde C = \frac{C-a}{\ell} \iff \mca(A) \leq C.\] 
Moreover 
\[ \wp [A;B] = \sum_{j\geq 0}  A(x_j) \left( \sum_{i<j} B(x_i) - \sum_{i>j} B(x_i) \right) \]  
\[ = \sum_{j\geq 0}  \widetilde A(j/M) \left( \sum_{i<j} \widetilde B(i/M) - \sum_{i>j} \widetilde B(i/M) \right) = \wp [ \widetilde A; \widetilde B].\] 
\end{proof}

\begin{remark} In all cases, the constraint value $C>a$.  For any finite collection of competing strategies there is $R>a$ so that their supports are all contained in $[a, R]$, as well as in $[a, L]$ for any $L>R$.  Without loss of generality, assume $R \geq 2C-a$.  By the preceding lemmas, this is equivalent to analyzing competition for strategies supported in $[0, 1]$ with constraint value $C \leq \frac 1 2$.  We will therefore make this assumption in the subsequent analyses.  
\end{remark}

\subsection{A sufficient condition for equilibrium strategies} \label{ss:reduce_to2}
Here we demonstrate a sufficient condition for a collection of strategies to be an equilibrium point.  This allows us to reduce to considering pairwise competition.  Once we complete the analysis for pairwise competition, we will prove that the sufficient condition is also necessary and thereby identify all equilibrium strategies.  We begin by computing for two competing strategies, 
\begin{align}
\wp [f;g] &= \int_0^1f(x)\left(2G(x)-G(1)\right) dx \nonumber \\
&=\left(F(1)G(1)-2\int_0^1 F(x)g(x)dx\right) \nonumber \\
&= \int_0^1g(x)\left(F(1)-2 F(x)\right)dx \nonumber \\ 
&=-\wp [g;f] \implies \wp [f;g] + \wp[g;f] = 0. \label{eq:zerosum}
\end{align}
This reflects the fact that each team collects all its winnings and pays all its losses to the competing teams, hence the total value across all teams remains constant.  One could interpret this as competition for a limited amount of resources.  We therefore have for a collection of competing teams 
\[ \sum_{k=1} ^n \wp(f_k; \ldots) = 0,\] 
where $\wp(f_k:\ldots)$ indicates the payoff to strategy $f_k$ competing against all others.  

As shown in \cite{thee, plos} 
\begin{equation} \label{eq:0sum} \wp[A;B] + \wp[B;A] = 0 \implies \wp[B;A] = - \wp[A;B], \end{equation}  
and similarly for a collection of competing teams, 
\[ \sum_{k=1} ^n \wp(A_k; \ldots) = 0.\] 
Above $\wp(A_k:\ldots)$ indicates the payoff to strategy $A_k$ competing against all others.

\begin{prop} \label{prop:many2pairs}
Assume that a collection of strategies $(f_1, \ldots, f_n)$ for the bounded measurable and continuous games satisfies 
\begin{equation} \label{eq:prop_many2pairs} \wp(f_k; f_j) = 0 \forall j, k, \quad \wp(f_k; g) \geq 0 \textrm{ for any strategy } g.  \end{equation} 
Then $(f_1, \ldots, f_n)$ is an equilibrium point. The analogous statement holds for the discrete game. 
\end{prop} 

\begin{proof} 
Assume that a collection of strategies satisfies \eqref{eq:prop_many2pairs}.  Then it follows by the definition of the payoffs that for all $k=1, \ldots, n$, 
\[ \wp(f_k; f_1, \ldots, f_{k-1}, f_{k+1}, \ldots, f_n) = 0.\] 
Moreover, by the zero sum dynamic, for any strategy $g$ we have $\wp(g; f_k) \leq 0$, and so again by the definition of the payoffs, for all $k=1, \ldots, n$, 
\[ \wp(g;  f_1, \ldots, f_{k-1}, f_{k+1}, \ldots, f_n) \leq 0 = \wp(f_k; f_1, \ldots, f_{k-1}, f_{k+1}, \ldots, f_n).\] 
This collection of strategies is therefore an equilibrium point.  The argument for the discrete game is identical. 
\end{proof} 

\subsection{The bounded measurable and continuous games of teams} \label{ss:bmc_C_half} 

\begin{prop}
If $C= \frac 1 2$, then a pair of functions that are positive and constant on $[0, 1]$ and zero elsewhere is an equilibrium point for the bounded measurable game of teams.  
\end{prop}
\begin{proof}
	For any $g$ competing with 
	\[ u(x) := \begin{cases} U(1) & 0 \leq x \leq 1 \\ 0 & 1 < x \end{cases} \] 
	
	\begin{multline*}
	\frac{1}{U(1) G(1)} \wp[u;g]=\int_0^1\frac{1}{G(1)}\left(\int_0^xg(t)dt-\int_x^1g(t)dt\right)dx \\ = \left(\int_0^1 2\frac{G(x)}{G(1)}dx-1\right) 
	= \left(1-\frac{2}{G(1)}\int_0^1xg(x)dx\right)=1-2\mca(g) \geq 0
	\end{multline*}
	The inequality follows from the constraint~(\ref{eq:constraint}).  Moreover, if $g$ is also positive and constant on $[0, 1]$ and zero outside this interval, then $\mca(g) = \frac 1 2$, and so we therefore have that $u$ and $g$ satisfy the necessary and sufficient conditions to be an equilibrium strategy. 	
	\end{proof}

\begin{prop}
	\label{prop:mca12}
	Let $f$ be a bounded measurable strategy subject to the constraint with $C= \frac 1 2$.  Assume that $f$ is not constant on $[0, 1]$, and that $f$ is supported in $[0, 1]$. Then there exists a bounded measurable strategy $g$ subject to the same constraint and supported on $[0, 1]$ for which $E[f;g]<0.$
\end{prop}

\begin{proof}
	If $\mca(f)<1/2$ then a strategy $g(x)$ that is positive and constant on $[0, 1]$ and supported in this interval satisfies $\wp[f;g] < 0$.   We may therefore henceforth assume $\mca(f)=1/2$.  Then, 
	\[ - \int_0 ^1 \frac{f(x)}{2} dx = - \frac{F(1)}{2} = - \int_0 ^1 x f(x) dx,\]
	and 
	\[ \wp[f; g] = \int_0 ^1 f(x) (2 G(x) - G(1)) dx = 2 G(1)  \int_0 ^1 f(x) \left( \frac{G(x)}{G(1)} - \frac 1 2 \right) dx\] 
	\[ \implies \wp[f;g] = 2 G(1)  \int_0 ^1 f(x) \left( \frac{G(x)}{G(1)} - x \right) dx,\] 
	and so similarly if $\mca(g) = 1/2$, we have 
	\[ \wp[g; f] = 2 F(1)  \int_0 ^1 g(x) \left( \frac{F(x)}{F(1)} - x \right) dx.\]  
	
	Since $f$ is not constant there exists $x\in[0,1]$ such that $F(x)\neq xF(1)$. Thus, the integrand above must assume both positive and negative values on sets of positive measure. We note that since $f \in \cL^\infty$, it follows that $F(x)$ is continuous.   Consequently, there is a non-empty open interval $(a,b) \subset [0,1]$, and a constant $R>0$, such that 
	\begin{equation}\label{eq:notconstant}
	 \frac{F(x)}{F(1)}-x>R\quad \forall \quad x\in[a,b].
	\end{equation}
	If $(a,b)$ is not fully contained in either $(0,1/2)$ or $(1/2,1)$, then we split $(a,b)$ into smaller intervals, one of which is fully contained in either $(0,1/2)$ or $(1/2,1)$.  We therefore assume without loss of generality that $(a,b)$ is contained in either $(0,1/2)$ or $(1/2,1)$. 
	First, assume $(a,b)\subset (1/2,1)$.
	Define for parameters $M, N, \delta$ to be chosen suitably, 
	\begin{equation} \label{eq:g_half_1} 
	g(x) = \begin{cases} M & x \in [0, \delta] \\ 
	N & x \in (a,b) \\
	0 & \mathrm{otherwise.} 
	\end{cases} 
	\end{equation} 
	Then 
	\[ G(1) = M \delta + N (b-a), \quad \mca(g) = \frac{M \delta^2/2 + N(b^2 - a^2)/2}{M\delta + N (b-a)}.\] 
	To guarantee that $\mca(g) = 1/2$, we therefore require
	\begin{align} 
	M \delta^2 + N(b^2 - a^2) &= M \delta + N (b-a)  \nonumber \\  
	& \iff M = \frac{N(b-a)(b+a-1)}{\delta(1-\delta)}.   \label{eq:mca_g_half_1} 
	\end{align} 
	Since $(a,b) \in (1/2, 1)$, $b+a-1>0$, and so it is possible to choose $M, N > 0$ and $0<\delta<1$ so that this equation is satisfied.  Then, 
	\[ \wp[g; f] = 2 F(1)  \int_0 ^1 g(x) \left( \frac{F(x)}{F(1)} - x \right) dx \] 
	\[ = 2 F(1)  \left[ \int_0 ^\delta M \left( \frac{F(x)}{F(1)} - x \right) dx + \int_a ^b N \left( \frac{F(x)}{F(1)} - x \right) dx \right]\] 
	\[ \geq 2 F(1) \left( - M \delta^2/2 + N R (b-a) \right) = F(1)  \left( - \frac{N (b-a)(b+a-1) \delta}{1-\delta} + 2 N R (b-a) \right)\] 
	\[ = F(1) N (b-a) \left( 2R - \frac{\delta (b+a-1)}{1-\delta} \right).\] 
	Assume we choose $\delta \in (0, 1/2)$.  Then since $b+a-1 \leq 1$, $\wp[g;f] > 0$ for any choice of $\delta \in (0, R) \cap (0, 1/2)$.  
	
If $(a,b)\subset (0,1/2)$, define
\begin{equation} \label{eq:g_ab_less} 
g(x) = \begin{cases} N & x \in (a,b) \\ M & x \in (1-\delta, 1] \\ 0 & \mathrm{otherwise} 
\end{cases} 
\end{equation} 
Then 
\[ G(1) = N (b-a) + M \delta, \quad \mca(g) = \frac{N (b^2-a^2)/2 + M(2\delta - \delta^2)/2}{N(b-a) + M\delta}.\] 
We therefore fix 
\[ \mca(g) = \frac 1 2 \implies N(b^2-a^2) + M(2\delta - \delta^2) = N(b-a) + M \delta \] 
\[\iff  M = \frac{N(b-a)(1-b-a)}{\delta (1-\delta)}. \] 
Noting that $(a,b) \subset (0,1/2)$ it is therefore possible to choose $N, M > 0$ and $\delta \in (0, 1)$ to satisfy this equality.  

Then, 
we estimate 
\[ \left| \frac{F(x)}{F(1)} - x \right| = \left| \frac{F(x) - F(1) + F(1)(1-x)}{F(1)} \right| \leq \frac{|F(x) - F(1)|}{F(1)} + |1-x| \] 
\[ \leq \left(\frac{||f||_\infty}{F(1)} +1\right)(1-x), \quad x \in (0, 1),\]
having used 
\[ |F(x) - F(1)| = F(1) - F(x) = \int_x ^1 f(x) dx \leq (1-x) ||f||_\infty.\] 
 Hence

	\[ \wp[g; f] = 2 F(1)  \int_0 ^1 g(x) \left( \frac{F(x)}{F(1)} - x \right) dx \] 
	\[ = 2 F(1)  \left[ \int_a ^b N \left( \frac{F(x)}{F(1)} - x \right) dx + \int_{1-\delta} ^1 M \left( \frac{F(x)}{F(1)} - x \right) dx \right]\] 
	\[ \geq 2 F(1)  \left(N R (b-a) - M \delta^2\left(1+\frac{||f||_\infty}{F(1)}\right) \right) \] 
	\[= 2 F(1)  \left( N R (b-a) - \frac{N(b-a)(1-b-a) \delta}{(1-\delta)}\left(1+\frac{||f||_\infty}{F(1)}\right)  \right)\] 
	\[ = 2 F(1) N (b-a) \left( R - \delta \frac{(1-b-a)}{1-\delta}\left(1+\frac{||f||_\infty}{F(1)}\right)  \right).\] 
	Since $(a,b) \subset (0, 1/2)$, $0<1-b-a<1$, so assuming that $\delta \in (0, 1/2)$, we obtain that $\wp[g;f] > 0$ for any 
	\[ \delta \in \left(0, \frac{R}{2(1+||f||_\infty/F(1))} \right) \bigcap (0, 1/2).\] 
\end{proof}

\begin{cor} \label{cor:eq_c_half} In case the constraint $C=1/2$, all equilibrium strategies in the bounded measurable game of teams for functions supported in $[0, 1]$ are positive constants on the unit interval, and conversely, all equilibrium points are comprised of positive constant functions.  
\end{cor} 

\begin{proof} A collection of equilibrium strategies $(f_1, \ldots, f_n)$ must satisfy 
\[ \wp[f_k;f_j]=0, \quad \forall j, k, \quad \wp[f_k; g] \geq 0\] 
for any other strategy $g$.  By the preceding proposition, the only strategies that satisfy these conditions are those in the statement of the corollary. 
\end{proof} 

\begin{cor} In the continuous game of teams, there are no equilibrium strategies for $C=1/2$. 
\end{cor} 
\begin{proof}  In \cite[Theorem 1]{plos}, we proved that all equilibrium strategies, if we consider \em only \em the unit interval, are constant positive functions.  However, the way we have defined strategies here, they are continuous on $[0, \infty)$.  Consequently, they can never be of this type.  
\end{proof} 

\begin{remark} There is no contradiction to \cite[Theorem 1]{plos}, because there we could ignore everything outside of the interval $[0, 1]$.  Here, we reduce the problem to this interval, but the strategies are assumed to be continuous on $[0, \infty)$.  Consequently, we cannot simply ignore everything outside of the interval.  In some regard, this shows that the bounded measurable game of teams is more natural. 
\end{remark}

We assume next that the constraint value is in $(0, 1/2)$.  

\begin{thm}
	\label{teo:eq_less_half}
	Assume that $C \in (0, 1/2)$.  Then all equilibrium strategies in the bounded measurable game for functions supported in $[0, 1]$ are comprised of elements of $\cL^\infty$ that are almost everywhere equal to 
		\[ \begin{cases} 0 & x \in [2C, 1] \\ a & x \in [0, 2C] \end{cases} \] 
for some positive constant $a$.  For the continuous game, there are no equilibrium strategies.  	
\end{thm}

\begin{proof}
We compute that the payoff 
\begin{multline*}
\wp[g;f]=  \int_0^1g(x)\left(\int_0^xf(t)dt -\int_x^1f(t)dt\right)=\\
= \int_0^1g(x)\left(2F(x)-F(1)\right)dx = 2F(1)\int_0^1g(x)\left(\frac{F(x)}{F(1)}-\frac{1}{2}\right)dx,
\end{multline*}
Then, insert
\[
\frac{1}{2}\int_0^1g(x)dx =\frac{1}{2}G(1)= \frac{1}{2\, \mca(g)}\int_0^1 xg(x)dx.
\]
We then find that
\begin{equation} \label{eq:equivalent_useful} 
\frac{1}{2F(1)}\wp[g;f]=\int_0^1 g(x)\left(\frac{F(x)}{F(1)}-\frac{x}{2\,\mca(g)}\right)dx,
\end{equation}
and 	
\begin{equation} \label{eq:equivalent_useful2} 
\frac{1}{2G(1)}\wp[f;g]=\int_0^1 f(x)\left(\frac{G(x)}{G(1)}-\frac{x}{2\,\mca(f)}\right)dx.
\end{equation}

	\textbf{Case 1 in the bounded measurable game:} Assume that $f$ is not identically zero on the interval $(2C, 1]$.  Since we are working in $\cL^\infty$, functions that differ on sets of measure zero are identical as elements of $\cL^\infty$, so it is equivalent to assume that $f$ is positive on a set of positive measure inside $(2C, 1]$. Consider the function:
	\begin{equation} \label{eq:h_bm} 
	h(x) := \begin{cases}
	1 & x\in[0,2C], \\
	0 & x\in(2C,1]. \\
	\end{cases}
	\end{equation}
	Then as it is defined $h$ is a strategy.  We note that for any strategy $g$, 
\[ 	\frac{1}{2G(1)}\wp[f;g]=\int_0^1 f(x)\left(\frac{G(x)}{G(1)}-\frac{x}{2\,\mca(f)}\right)dx \leq \int_0^1 f(x)\left(\frac{G(x)}{G(1)}-\frac{x}{2\,C}\right)dx,\] 
because 
\[\mca(f) \leq C \implies \frac{x}{\mca(f)} \geq \frac x C.\] 
	
	We then compute that for
	\begin{equation} \label{eq:bigH} H (x) := \int_0 ^x h (t) dt,
	\end{equation}  
	\[ \frac{1}{2H(1)}\wp[f;h] \leq \int_0^1 f(x)\left(\frac{H(x)}{H(1)}-\frac{x}{2\,C}\right)dx\] 
	\[ =\int_{2C}^{1} f(x) \left( \frac{H(x)}{H(1)} - \frac{x}{2C}\right) dx  = \int_{2C}^{1} f(x) \left( 1 - \frac{x}{2C}\right) dx < 0,\] 
	since, $\left( 1 - \frac{x}{2C}\right) < 0$ on $(2C,1]$ and $f$ is non-zero on $(2C,1]$ and strictly positive on a set of positive measure by assumption.

	\textbf{Case 2 in the bounded measurable game:} Assume that $f(x) = 0$ for $x\in [2C,1]$ (equivalently, $f=0$ almost everywhere in $[2C,1]$ but since we work in $\cL^\infty$ this is equivalent).  The proof in this case then reduces to the case in which the constraint value is $1/2$ and the competitive abilities are selected from the range $[0, 1]$ by Lemma \ref{le:translation_bmc}.

	\textbf{Case 1 in the continuous game:} Assume that $f$ is not identically zero on the interval $(2C, 1]$.  In this case we shall begin with a bounded measurable function that is discontinuous and approximate it by continuous functions.  For the functions $h$ and $H$ defined in \eqref{eq:h_bm} and \eqref{eq:bigH}, respectively, we have computed in Case 1 of the bounded measurable game that $\wp[f; h] < 0$.  Since we require continuous functions, we define 
	\begin{equation*}
	h_{\varepsilon}(x) := \begin{cases}
	\frac{x}{\varepsilon} & x\in[0,\varepsilon], \\
	1 & x\in[\varepsilon, 2C-\varepsilon], \\
	-\frac{x}{\varepsilon} + \frac{2C}{\varepsilon}, & x\in [2C-\varepsilon, 2C],\\
	0 & x\in [2C,1].
	\end{cases}
	\end{equation*}
	Let us note that $h_{\varepsilon}$ is a continuous non-negative function, and $\mca(h_{\varepsilon}) = C$ since it is symmetric with respect to $C$ on $[0,2C]$ and identically zero on $[2C,1]$.
	
	Next, we note that $h_{\varepsilon} \rightarrow h$ pointwise almost everywhere on $[0,1]$ as $\varepsilon\rightarrow0$, therefore,
	\begin{equation*}
	f(x) \left( \frac{H_{\varepsilon}(x)}{H_{\varepsilon}(1)} - \frac{x}{2C}\right) \rightarrow f(x) \left( \frac{H(x)}{H(1)} - \frac{x}{2C}\right)
	\end{equation*}
	pointwise almost everywhere on $[0,1]$ as $\varepsilon\rightarrow0$. The dominated convergence theorem and $\wp[f; h]<0$ imply that
	\begin{equation*}
	\int_{0}^{1} f(x) \left( \frac{H_{\varepsilon}(x)}{H_{\varepsilon}(1)} - \frac{x}{2C}\right) dx < 0 \implies \wp[f; h_\varepsilon]<0
	\end{equation*}
	for sufficiently small $\varepsilon>0$.
\\

	\textbf{Case 2 in the continuous game:} Assume that $f(x) = 0$ for $x\in [2C,1]$. We then consider
	\[ \widetilde f(t) := f(2Ct), \quad t \in [0,1].\] 
	Since $f$ is continuous on $[0,1]$, $f$ cannot be a positive constant on $[0,2C]$, and therefore $f(2Ct)$ is not equal to a positive constant for $t \in [0, 1]$.  The proof in this case follows from Lemma \ref{le:translation_bmc} and Theorem 1 in \cite{plos}.

	\end{proof}

\begin{prop} \label{prop:2many_reverse_bmc}
Assume that a collection of strategies $(f_1, \ldots, f_n)$ for the bounded measurable or continuous game is an equilibrium point.  Then they satisfy 
\[\wp(f_k; f_j) = 0 \forall j, k, \quad \wp(f_k; g) \geq 0 \textrm{ for any strategy } g.\] 
Equivalently, each of $f_k$ is an equilibrium strategy for the two-player game. 
\end{prop} 

\begin{proof} 
Assume that $(f_1, \ldots, f_n)$ is an equilibrium point.  Then by definition of equilibrium point 
\[ \wp(f_k; f_1, \ldots, f_{k-1}, f_{k+1}, \ldots, f_n) \geq \wp(\sum_{\ell \neq k} f_\ell;  f_1, \ldots, f_{k-1}, f_{k+1}, \ldots, f_n) =0.\] 
Above, $\sum_{\ell \neq k} f_\ell$ is the strategy obtained by summing the strategies $f_\ell$ for all $\ell \neq k$.  Note that this is also a strategy.  
The above inequality holds for all $k=1, \ldots, n$.  By the zero sum dynamic 
\[ \sum_{k=1} ^n \wp(f_k; f_1, \ldots, f_{k-1}, f_{k+1}, \ldots, f_n) =0. \] 
Since each summand is non-negative, they must all vanish.  Consequently, for any strategy $g$, by definition of equilibrium strategy, 
\[ \wp(g; f_1, \ldots, f_{k-1}, f_{k+1}, \ldots, f_n) \leq 0 = \wp(f_k;  f_1, \ldots, f_{k-1}, f_{k+1}, \ldots, f_n),\] 
and this holds for all $k=1, \ldots, n$.  We therefore have for the particular choice $g=f_j$ for some fixed $j$ that 
\[ \wp(f_j; \sum_{\ell \neq k} f_\ell) \leq 0 \quad \forall k, \implies \sum_{k=1} ^n \wp(f_j; \sum_{\ell \neq k} f_\ell) \leq 0\] 
\[ \implies \wp(f_j; f_1, \ldots, f_{j-1}, f_{j+1}, \ldots, f_n) \leq 0.\] 
Above we have used the definition of the payoffs and the linearity of integration.  Since this last inequality has been shown to be an equality, each of the summands must vanish, showing that 
\[ \wp(f_j; \sum_{\ell \neq k} f_\ell) = 0, \quad \forall j, k.\] 
Combining 
\[ \wp(f_j; f_1, \ldots, f_{j-1}, f_{j+1}, \ldots, f_n) = 0 \textrm{ and } \wp(f_j; \sum_{\ell \neq k} f_\ell) = 0 \implies \wp(f_j; f_k) =0,\] 
for all $j$ and $k$.  Consider for a moment the case in which there are only two competing strategies.  Then a necessary and sufficient condition for $(f,h)$ to be an equilibrium point is that 
\begin{equation} \label{eq:two_fhg} \wp(f; h) =0 = \wp(h; f), \quad \wp(f; g) \geq 0, \quad \wp(h; g) \geq 0, \quad \textrm{for all strategies $g$.} \end{equation} 
Moreover, having identified all equilibrium strategies in the two player game, it follows that if $f-h \geq 0$ and is not identically zero, then $f-h$ is also an equilibrium strategy.  Define 
\[ g_k := \sum_{\ell \neq k} f_\ell \implies \wp(g; g_k) \leq \wp(f_k; g_k) = 0 \quad \textrm{for all strategies $g$,}\] 
and 
\[ \wp(g_k; g_j) = 0 \quad \forall j, k.\] 
Consequently, each $g_k$ is an equilibrium strategy for the two player game.  By linearity, the sum of two equilibrium strategies is again an equilibrium strategies.  We therefore have 
\[ \sum_{k=1} ^n g_k = (n-1) \sum_{k=1} ^n f_k \] 
is an equilibrium strategy.  A nonzero scalar multiple of an equilibrium strategy is again an equilibrium strategy by linearity, hence 
\[ \sum_{k=1} ^n f_k \] 
is an equilibrium strategy.  Then since $f_k$ is not identically zero by definition of strategy, 
\[ \sum_{k=1} ^n f_k - g_k = f_k \] 
is an equilibrium strategy for the two player game, for each $k=1, \ldots, n$. 
\end{proof}


\subsection{The discrete game of teams} \label{s:discrete} 
In the case where the constraint value $C=1/2$, we have found all equilibrium strategies in \cite[Theorem 1]{plos}.  We summarize the results obtained in \cite{plos} that determine all equilibrium strategies for the constraint value $C=1/2$.  

\begin{thm}[Theorem 1 of \cite{plos}] \label{thm:mca_half} In case $M$ is odd, and the constraint value $C=\frac 1 2$, then all equilibrium strategies supported in 
\[ \left \{ \frac j M \right\}_{j=0} ^M\] 
are  \em uniform \em strategies.  A uniform strategy $U$ satisfies $U(x_i) = a $ for all $0 \leq i \leq M$, for some constant $a>0$. 
In case $M$ is even, then all equilibrium strategies are those $A$ which have $|A| > 0$, $\mca(A) = \frac 1 2$ and furthermore satisfy
\[ A (x_{2j}) = A (x_0), \quad A (x_{2j+1}) = A (x_1), \quad \forall j\in\{0,1,...,M/2\}. \] 
\end{thm} 
 
Next we assume that the constraint value $C<1/2$ and is similar to the case in which $C=1/2$, namely the constraint value satisfies 
\[ C= \frac{j_C}{M} \textrm{ or } C = \frac{2j_C + 1}{2M}.\] 
The equilibrium strategies in this case are of two types, analogous to those in the case when $C=1/2$ and $M$ is either odd or even.  

\begin{thm}
	 If the $\mca$ constraint is for $C=\frac{j_C}{M}<1/2$ then all equilibrium strategies are those with $\mca = C$ that are of the form 
	\[ A(x_{2k}) = \begin{cases} a, & 0 \leq k \leq j_C, \\ 0 & k \geq j_C + 1, \end{cases} \quad A(x_{2k+1}) = \begin{cases} b, & 0 \leq k \leq j_C - 1, \\ 0, & k \geq j_C. \end{cases} \] 
If the $\mca$ constraint is for $C= \frac{2j_C + 1}{2M}<1/2$, then all equilibrium strategies are of the form 
\[ A(x_k) = \begin{cases} c, & 0 \leq k \leq 2 j_C + 1, \\ 0, & k \geq 2j_C + 2,\end{cases} \] 
for any constant $c>0$.  
\end{thm}

\begin{proof}
	If $C=j_C/M,$ we define 
	\[
	A(x_k):=
	\begin{cases}
		1, & k\in\{0,1,2,...,2j_C\}\\
		0, & k>2j_C.
	\end{cases}
	\]
	Then $|A|=2j_C+1$, and $\mca(A)=C$. We compute $\wp[A;B]$ for competition against a strategy $B$ subject to the same constraint 
	\begin{multline*}  
		\wp[A;B] =\sum_{k=0}^{2j_C}\left(\sum_{i=0}^{k-1}B(x_i)-\sum_{i=k+1}^{M}B(x_i)\right)
		\\=\sum_{k=0}^{2j_C}\left(2\sum_{i=0}^{k-1}B(x_i)+B(x_k)-|B|\right)
		\\= 2 \sum_{k=0} ^{2j_C} \sum_{i=0} ^{k-1} B (x_i) + \sum_{k=0} ^{2 j_C} B (x_k) - (2j_C + 1) |B| 
		\\ = 2 \sum_{k=0} ^{2j_C} (2j_C - k) B (x_k) + \sum_{k=0} ^{2 j_C} B (x_k) - (2j_C + 1) |B| 
\\= 2\sum_{k=0}^{2j_C}(j_C-k)B(x_k)+(2j_C+1)\sum_{k=0}^{2j_C}B(x_k)-(2j_C+1)\sum_{k=0}^{M}B(x_k)
		\\=2\sum_{k=0}^{2j_C}(j_C-k)B(x_k)-(2j_C+1)\sum_{k=2j_C+1}^{M}B(x_k).
	\end{multline*}
	Since $B$ is subject to the constraint 
	\[
	\mca(B)=\frac{\sum_{k=0}^M\frac{k}{M}B(x_k)}{\sum_{k=0}^M B(x_k)}\leq C=\frac{j_C}{M} \iff \sum_{k=0} ^M (k-j_C) B (x_k) \leq 0
	\]
	it follows that 
	\begin{equation}
		\label{mca-in-proof-1}
		\sum_{k=2j_C+1}^M(k-j_C)B(x_k)\leq \sum_{k=0}^{2j_C}(j_C-k)B(x_k).
	\end{equation}
	Thus, 
	\begin{align} \nonumber  
		\wp[A;B]\geq  2\sum_{k=2j_C+1}^{M}(k-j_C)B(x_k) - (2j_C+1)\sum_{k=2j_C+1}^{M}B(x_k) 
		\\ 
		 \label{eq:C_less_half_proof_end}   =\sum_{k=2j_C+1}^{M}(2k-4j_C-1)B(x_k) \geq 0. 
	\end{align}
	Since $2k-4j_C - 1 \geq 0$ for all $k \geq 2j_C + 1$, and $B(x_k) \geq 0$ for all $k$, equality holds in \eqref{eq:C_less_half_proof_end}  if and only if $B(x_k) = 0$ for all $k > 2j_C$.  If this is not the case, then $\wp[A;B] > 0$, and therefore $A$ defeats $B$.  We note that the same holds for any other team that has positive identical values at $x_k$ for $k=0, \ldots, 2j_C$ and zero at all other $x_k$.  Consequently, it suffices to consider the problem for the interval $[0, 2j_C/M]$ with competitive abilities 
	\[ 0 < \frac 1 M < \ldots < \frac{2j_C}{M}, \quad C = \frac{j_C}{M}.\] 
	As shown in Lemma \ref{le:translation_disc}, this problem is equivalent to the case in which the constraint is equal to $\frac 1 2$, and $M$ is even.

	Now assume that the constraint value is of the second type.  The MCA constraint admits the following reformulation:
	\[
	\sum_{k=0}^M k B(x_k)\leq \left(j_C+\frac{1}{2}\right)\sum_{k=0}^M B(x_k) \iff \sum_{k=0} ^M (k - j_C - 1/2) B (x_k) \leq 0.
	\]
	Thus,

	\begin{equation}
		\label{mca-second-case}
	\sum_{k=2j_C + 2} ^M (k-j_C-1/2) B (x_k) \leq \sum_{k=0} ^{2j_C+1} (j_C + 1/2 - k) B (x_k).	
	\end{equation}
	We define the strategy $A$ such that 
	\[
	A(x_k)=
	\begin{cases}
		1, & k\in\{0,1,2,...,2j_C+1\}\\
		0, & k>2j_C+1.
	\end{cases}
	\]
	Then $|A| = 2j_C + 2$, and $\mca(A) = C = \frac{j_C}{M} + \frac{1}{2M}$.  
	The payoff
	\begin{multline*} 
		\wp[A;B]=\sum_{k=0}^{2j_C+1}\left(2\sum_{i=0}^{k-1} B(x_i)+ B(x_k)-\sum_{i=0}^{M} B(x_i)\right)\\
		=2\sum_{k=0}^{2j_C+1}(2j_C+1-k) B(x_k)+\sum_{k=0}^{2j_C+1} B(x_k) -(2j_C+2)\sum_{k=0}^{M} B(x_k)\\
		=\sum_{k=0}^{2j_C+1}(2j_C+1-2k) B(x_k)-(2j_C+2)\sum_{k=2j_C+2}^{M} B(x_k).
	\end{multline*}
	Using the MCA constraint and~(\ref{mca-second-case}), we find
	
		\begin{align*} \wp[A;B] &
		\geq \sum_{k=2j_C+2}^M(2k-1-2j_C) B(x_k)-(2j_C+2)\sum_{k=2j_C+2}^{M} B(x_k)\\
		&=\sum_{k=2j_C+2}^M(2k-3-4j_C) B(x_k)\geq 0.
	\end{align*}
	Above, we use the facts that $B(x_k) \geq 0$ for all $k$, and $(2k-3-4j_C) > 0$ for $k \geq 2j_C + 2$.  Hence each term in the sum is non-negative, and the inequality is an equality if and only if $B (x_k) = 0$ for all $k \geq 2j_C+2$.  It therefore suffices to consider teams with competitive abilities contained in the range $[0, \frac{2j_C+1}{M}]$, subject to the constraint $\mca \leq C = \frac{j_C}{M} + \frac{1}{2M}$.  By the translation invariance of the problem as demonstrated in Lemma \ref{le:translation_disc}, this is equivalent to the case in which $C=\frac 1 2$, and $M$ is odd.  
	
\end{proof}

\begin{prop} \label{prop:2many_reverse_d}
Assume that a collection of strategies $(A_1, \ldots, A_n)$ for the discrete is an equilibrium point.  Then they satisfy 
\[\wp(A_k; A_j) = 0 \forall j, k, \quad \wp(A_k; B) \geq 0 \textrm{ for any strategy } B.\] 
Equivalently, each of $A_k$ is an equilibrium strategy for the two-player game. 
\end{prop} 

\begin{proof} 
The proof is obtained from the proof for the bounded measurable and continuous games by an identical argument, by substituting $A_k$ for $f_k$ and $B$ for $g$.  
\end{proof}

\section{Discussion}  \label{s:discussion} 
Games involving competing teams are widely researched and applied in numerous contexts; see \cite{peterson2000multiplayer} and references therein. Team games differ significantly from the player-to-player games that prevailed the early days of game theory \cite{morgenstern1944, nash1950thesis}.  Teams are natural constellations in cooperative game theory, given that the value of the game increases if individuals group \cite[Chapter 8]{mazalov2014mathematical}. According to \cite{hiller2019structure}, cooperative game theory models ``the combination of specialized expertise within the team.''  However, in the aforementioned work and many other studies based on cooperative game theory, teams do not necessarily compete with other teams \cite{hiller2019structure,hernandez2010rankings,semsar2009multi}.  Our teams may have different and dynamical sizes, which is a major difference to \cite{hernandez2010rankings}, who defined a team game as ``a cooperative game in TU-form, whose values on coalitions of every cardinality but one are zero.'' 

Many authors investigate teams with non-cooperative game theory.  However, there too, it is common to form teams without any actual competition between them.  Examples include a selection process for team formations within a single sports club \cite{chambers2017non}, in governance \cite{han2013analysis}, or  technology investment \cite{liao2009decision}.  In \cite{liu2004noninferior} they propose a hybrid approach that bears some resemblance to ours.  Their teams are collections of individuals with cost functions that depend on the actions of all players -- including those of other teams. Considering a pair of teams, they ``stipulate that the relationship between the two teams is completely adversarial and that cooperation between them is not permissible.  In other words, both cooperation within each team and competition between the teams must coexist.''  Conflicting teams of cooperating players were also studied in \cite{saad2010}.  However, as noted in the review \cite{liang2012game}, there is a general lack of multi-player games in conflict descriptions as most authors model conflicting agents as single players or assume that a conflict is a 2-player multi-stage interaction. 

Although teams refer to a constellation of individuals, our results have applications far broader, because the individuals comprising a team are an abstract concept capable of representing anything. For example, the individuals comprising a team could represent investment products within a financial portfolio.  The portfolio may be constructed in many ways, and it is natural to impose budget constraints, for example on the total value of the portfolio \cite{veysoglu2002thesis}.  Portfolio optimization uses game theory, both cooperative \cite{gambarelli2004takeover,chis2017coalitional,simonian2019portfolio} and non-cooperative \cite{amihud1974portfolio,bell1988game,young1998minimax,yang2013multi,fu2017information}. 
Teams also are important in evolutionary game theory, in which a standard approach is non-cooperative game theory \cite{axelrod1981evolution,riechert1983game}.  The motivation for non-cooperative game theory is that the individuals comprising a team act independently.  In numerous contexts, this is a reasonable assumption.  For example, when the individuals in a team represent people, animals, microbes, or other organisms, most spontaneous decisions are made without consulting others.   If one considers the game as an aggregate over numerous decisions and subsequent consequences, then the majority of the actions taken by an individual are taken without consulting others.  Although this may not be perfectly accurate, a similar assumption is made in modern portfolio theory \cite{markowitz1952, markowitz1959}, by assuming that the prices of distinct investment products are independent.  This is not quite correct; it is a simplification that allows one to draw conclusions using the law of large numbers.  The prices of investment products can be and often are correlated.  Nonetheless, in spite of this imperfect simplification, modern portfolio theory remains widely in use today, indicating the utility of the theory, even if it is not perfect.  

Our first main contribution is a basic game theoretic framework for further analysis of the internal composition of a team and the repercussions for the team as a whole in competition with other teams.  The second main contribution is the rigorous identification of all equilibrium points and strategies.  These strategies correspond to the most heterogeneous team composition.  This indicates that a diverse team is a strong team in the face of competition with other teams.  It is important to note that in our model, the individuals in the teams are randomly paired to compete, implying a certain unpredictability.  Consequently, a diverse team may be a strong team in the face of new or unpredictable challenges, but it may not necessarily be the strongest team to face a specific, predictable challenge.  For example, in biology, there may be a single strategy that is best suited in one specific circumstance, known as an evolutionary stable strategy (ess) \cite{ess, algal_games}.  Similarly, in the face of one particular, constant challenge, a homogeneous team comprised of individuals characterized by an ess may be the strongest.  This work and our model is not in contradiction, but rather, complementary, since it applies to situations where the team faces unpredictable, possibly changing and new challenges.  

\section{Conclusion} \label{s:conclusion} 
Seeking a theoretical explanation for the strength of diversity within a team, one could argue that a team is a type of biological system, whether the team represents a collection of people, animals, organisms, or investment products.  According to \cite{biologyfirstlaw}  ``there exists in evolution a spontaneous tendency toward increased diversity and complexity, one that acts whether natural selection is present or not.''  The authors dubbed this the \em Zero-Force Evolutionary Law, \em or more colloquially, \em Biology's First Law.  \em  Identifying diversity and complexity with entropy,  the Second Law of Thermodynamics states that entropy in a physical system never decreases, it either remains constant or increases.  Consequently, both theoretical biology and physics  suggest a tendency towards non-decreasing diversity and complexity.  Our results give a novel mechanistic underpinning for the strength of diversity that is broadly applicable due to its foundation in theoretical mathematics and that is consistent with the predictions of the fundamentals laws of biology and physics.  

\section*{Acknowledgements} 
The authors thank Jil Kl\"under and Susanne Menden-Deuer for constructive criticism on a preliminary draft and for productive discussions.  JR and CJK are supported by the Swedish Research Council Grant 2018-03873 for which we extend our gratitude.  MN graciously acknowledges the support of the Ministry of Education and Science of the Republic of Kazakhstan grant AP09260223 and the Government of Kazakhstan and the World Bank grant APP-PHD-A-18/013P financed by the project “Fostering productive innovation.’’ JR thanks the National Science Foundation for the award DMS-1440140 which funded a semester at the Mathematical Sciences Research Institute.

\end{document}